\documentclass{article}[11pt]
\usepackage{caption}
\usepackage{subfigure}
\usepackage{amsfonts}
\usepackage{amsmath}
\usepackage{color}
\usepackage{indentfirst,latexsym,bm}
\usepackage{amssymb}
\usepackage{graphics}
\usepackage{graphicx}
\usepackage{endnotes}

\numberwithin{equation}{section}

\newcommand{\be}{\begin{eqnarray}}
\newcommand{\ee}{\end{eqnarray}}
\newcommand{\bee}{\begin{eqnarray*}}
\newcommand{\eee}{\end{eqnarray*}}

\newtheorem{theorem}{Theorem}[section]

\newtheorem{lemma}[theorem]{Lemma}

\newtheorem{remark}[theorem]{Remark}

\newtheorem{assumption}[theorem]{Assumption}
\newenvironment{proof}[1][Proof]{\textbf{#1.} }{\ \rule{0.5em}{0.5em}}

\begin{document}

\title{\LARGE\bf Pricing Model for Data Assets in Investment-Consumption Framework with Ambiguity
\footnote{\noindent The work is supported by NNSF of China (No. 12171169, 12371470, 12271188).}}

\date{}

\author{
Xiaoshan Chen~\thanks{School of Mathematical Science, South China Normal University, Guangzhou {\rm510631}, China, \texttt{xschen@m.scnu.edu.cn}.}$\quad$
Chen Yang~\thanks{Department of Systems Engineering and Engineering Management,The Chinese University of Hong Kong, New Territories, Hong Kong, \texttt{cyang@se.cuhk.edu.hk}.}$\quad$
Zhou Yang~\thanks{School of Mathematical Science, South China Normal University, Guangzhou {\rm510631}, China, \texttt{yangzhou@scnu.edu.cn}.}$\quad$}
\maketitle

\begin{abstract}
Data assets are data commodities that have been processed, produced, priced, and traded based on actual demand. Reasonable pricing mechanism for data assets is essential for developing the data market and realizing their value. Most existing literature approaches data asset pricing from the seller's perspective, focusing on data properties and collection costs, however, research from the buyer's perspective remains scarce. This gap stems from the nature of data assets: their value lies not in direct revenue generation but in providing informational advantages that enable enhanced decision-making and excess returns. This paper addresses this gap by developing a pricing model based on the informational value of data assets from the buyer's perspective. We determine data asset prices through an implicit function derived from the value functions in two robust investment-consumption problems under ambiguity markets via the indifference pricing principle. By the existing research results, we simplify the value function, using mathematical analysis and differential equation theory, we derive general expressions for data assets price and explore their properties under various conditions. Furthermore, we derive the explicit pricing formulas for specific scenarios and provide numerical illustration to describe how to use our pricing model.\\


\end{abstract}

\baselineskip 18 pt
\section{Introduction}

With the rapid development of methods such as statistical methods, machine learning algorithms, artificial intelligence systems, and large language models, the value of data has become more and more significant, and it plays an indispensable role in various industries and practical applications \cite{FR}. Although the data may contain extremely high potential value, the value of raw data is often difficult to realize, and its utilization efficiency is relatively low. To realize the true value of data, raw data must undergo cleansing, processing, and handling based on  the customer's needs and specific application to form effective data goods, these products then become data assets through pricing and trading. Clearly the pricing of data assets is a critical component of this process, a rational pricing mechanism for data assets can effectively promote the development of the data market and realize the potential of data value.

The actual price of the data assets is the result of a game-theoretic equilibrium between buyers and sellers, who weigh their respective circumstances to determine an acceptable price range. This highlights the critical importance of considering both the buyer's and seller's perspectives in pricing. Most existing literature analyzes pricing strategies based on data characteristics and production costs from the seller's point. However, research from the buyer's perspective remains scarce due to inherent challenges: data assets typically do not generate direct cash flow for users, they instead provide informational value that aids decision-making to achieve  excess returns or utility via enhancing investor's understanding. These indirect benefits are difficult to observe and quantify, creating fundamental obstacles for buyer-centric pricing research.

In this paper, from the buyer's perspective, applying the indifference pricing principle, we establish a mathematical pricing model for data assets in the optimal investment-consumption framework, focusing on the core value of data assets -- information value. In this model, we assume the buyer acts as an investment-consumption decision-maker who estimates the market parameters based on known information. The amount of information determines the ambiguity of the parameters. Under constraints on investment-consumption strategies, loan-deposit interest rate spreads and market ambiguity, the decision-maker makes investment-consumption decisions to maximize her robust expected utility. After purchasing the data asset, although the initial wealth decreases, the decision-maker gains more information, leading to narrower parameter estimation ranges, thereby achieving excess utility. The reasonable price should ensure that the maximum robust expected utility after purchasing the data asset is no less than the utility without the purchase, which is the maximum price that the buyer is willing to pay for the data asset. The price for data asset is determined implicitly through the value functions of two control-constrained game problems. We derive price descriptions and analyze their properties under general conditions,  while providing explicit calculation expressions under specific conditions.

We now review the relevant literature. Research on data pricing can be traced back to \cite{JBMC} in 2008, which explored data valuation from the perspective of data privacy. Subsequent studies, such as \cite{LR} and \cite{ZBL}, explored pricing mechanisms based on data acquisition, storage, and maintenance costs. These works exclusively adopted the seller's perspective, deriving data asset prices from costs and intrinsic properties.

Academic research on data asset pricing emerged later. \cite{V} first systematically addressed this problem, emphasizing that the core value of data assets lies in their informational value. The paper introduced multiple pricing frameworks including cost approach, revenue approach, value function estimation, and complementary input analysis. \cite{FSVV} further organized existing pricing methodologies. While some approaches in them considered informational value from the buyer's perspective, they only proposed static pricing concepts rather than establishing rigorous mathematical models.

The work most closely related to ours is \cite{JS}, which pioneered dynamic information pricing from the buyer's perspective by employing the indifference pricing principle and optimal portfolio choice model. It posited that investors, prior to acquiring information, could only perceive the expected return as a random variable following a known normal distribution. However, upon purchasing the information, they would gain knowledge of its exact value and obtain excess expected utility. The fair information price was thus defined as the value that equalizes the expected utilities of optimal investment strategies before and after information acquisition. Our model diverges from this framework in three key aspects: First, we quantify information value through the magnitude of parameter ambiguity regions. This facilitates the incorporation of investment-consumption models with ambiguity, explicitly enhancing their robustness. Second, our framework operates under more realistic and generalized market conditions. Specifically, we consider a multi-asset economy with consumption capabilities, incorporate investment-consumption constraints, account for information-enhanced volatility estimation, and introduce borrowing-lending interest rate spreads. Finally, our analytical approach yields remarkably concise closed-form solutions (as detailed in Theorem \ref{th-la=0}).

Since the investment-consumption problem in ambiguous markets is the core of our pricing model,  a literature review on this topic is essential. Research on financial issues with ambiguity dates back to 2002 with \cite{Chen2002}, which investigated optimal investment and pricing under ambiguity in the expected return of risky assets. Subsequently, \cite{Epstein} extended the research to scenarios with ambiguous volatility. Since then, academia has conducted in-depth studies on optimal investment models under probability or parameter ambiguity, yielding substantial achievements as seen in works such as \cite{BP, BCZ, Denis, DV, Epstein2, GI, MPZ, PWZ}, among others. In our previous work in \cite{YLZ}, we examined the robust investment-consumption problem under constant relative risk aversion utility, considering ambiguity in both expected returns and volatility parameters, along with investment-consumption constraints and borrowing spreads. Later in \cite{Yang2023}, we first utilized the size of ambiguous parameter regions to quantify information availability, building an optimal investment-consumption model with information costs. Inspired by \cite{JS}, this paper employs the indifference pricing principle and the investment-consumption model in ambiguous markets to address data asset pricing from the buyer's perspective.

The main contribution of this work is three-fold. First, we pioneer the use of the size of parameter ambiguity regions to quantify the value of data assets in the dynamic portfolio selection problems, enabling the study of data asset pricing from the buyer's perspective through the indifference pricing principle. This methodology aligns with practical scenarios where the primary value of data resides in parameter estimation. Simultaneously, it integrates classical robust control method into data asset pricing to enhance robustness, while expanding the application scope of financial models under ambiguity.

Second, we derive an exceptionally concise pricing formula within a highly general framework, which holds potential to serve as a foundational basis for further research on equilibrium price of data assets. Our model assumes a market with multiple risky assets whose estimated expected returns and volatilities can be influenced by data assets, alongside investment-consumption constraints and potential borrowing-lending spreads. The resulting pricing formula is strikingly simple: for the case without consumption, see Theorem \ref{th-la=0}; for cases incorporating consumption, refer to Theorems \ref{th-lo1} and \ref{th-po1}. Notably, Theorem \ref{th-la=0} demonstrates that under the assumption of constant relative risk aversion (CRRA) utility, the data asset price is proportional to initial wealth and determined by the {\it Investment Opportunity Index} before and after acquiring the data asset. This index synthesizes the impacts on investment opportunities generated by the market parameters with ambiguity, investment constraints, and the decision-maker's risk aversion coefficient (which is elaborated in Remark \ref{rhola0}).

Third, we rigorously analyze the time-dependent monotonicity of data asset prices in consumption-inclusive settings (see Theorems \ref{th-lo1} and \ref{th-po1}). Despite obtaining closed-form price expressions, their complexity renders direct monotonicity analysis infeasible. By combining analytical methods with the comparison theorems for differential equation , we establish distinct monotonicity patterns under varying parameter regimes. These findings reveal intriguing behavior: when the utility discount rate is low, prices exhibit monotonic decay over time; conversely, with sufficiently high discount rates, prices first increase and subsequently decrease. The financial interpretation of this phenomenon is elaborated in Remark \ref{rhola1}.

The rest of the paper is organized as follows. We will construct the pricing model for the data asset in Section 2. In Section 3, we first show the existence and uniqueness of the data price, then we will show the explicit formula of the data asset price under independent investment and consumption constraints, analyze the properties to learn how the consumption complicates the analysis of the data price. We will discuss more specific examples to show the explicit calculation formula of the Investment Opportunity Index and present more concrete results in Section 4. Section 5 numerically illustrates our data asset pricing method. Section 6 concludes.

\section{Pricing model formulation}

In this section, we develop a data asset pricing model by synthesizing the optimal investment-consumption framework in ambiguous markets with the indifference pricing principle. We assume the data asset's buyer maximizes robust expected utility through choosing proper constrained investment-consumption strategies in the finance market with parameter ambiguity. Moreover, by acquiring the data asset, she obtains additional information to mitigate market parameter ambiguity, thereby securing enhanced robust expected utility. The model is structured as follows: Subsection 2.1 specifies the financial market assumptions, Subsection 2.2 formalizes investment-consumption constraints, Subsection 2.3 constructs the robust utility maximization framework, and Subsection 2.4 derives the pricing model via the indifference principles.

\subsection{Financial market hypothesis}

Let $W:=(W_1,W_2,\cdots,W_m)^{\rm T}$ be an $m$-dimensional standard Brownian motion on a filtered probability space $(\Omega,{\cal F},\mathbb{F}:=\{{\cal F}_t:t\geq 0\},\mathbb{P})$, where $\mathbb{F}$ is the augmented filtration generated by $W$ and satisfies the usual conditions. There are a riskless bank account $B$ and $n$ risky assets $S:=(S_1,S_2,\cdots,S_n)^{\rm T}$ in the financial market.  The price processes of the risky assets $S:=(S_1,S_2,\cdots,S_n)^{\rm T}$ satisfy the following stochastic differential equations (SDE),
\be\label{riskyasset}
 \mathrm{d}S_{i,s} = \mu_{i,s} S_{i,s} \mathrm{d}s + \sum_{j=1}^m\sigma_{ij,s} S_{i,s}\mathrm{d}W_{j,s},
\ee
where  $\mu:=(\mu_1,\mu_2,\cdots,\mu_n)^{\rm T}$ and $\Sigma:=(\sigma_{ij})_{n\times m}$ represent the drift and volatility of the risky assets.

Regarding the bank account $B$, we assume that the borrowing rate $R$ and the lending rate $r$ are different constants. If $B$ is positive, the investor lends with rate $r$, if $B$ is negative, the investor borrows with rate $R$, it is natural to assume that $R\geq r$. Consequently, the bank account $B$ follows
\be\label{riskfreeasset}
 {\rm d}B_s=(r B_s^+- R B_s^-)\,{\rm d}s=\big(r B_s- (R-r)B_s^-\big)\,{\rm d}s,
\ee
where $x^+=\max(0,x)$, and $x^-=\max(0,-x)$.

Consider the investor trades both the risky assets and riskless asset, yet he has limited information about the risky assets' parameters $(\mu,\Sigma)$. The uncertainty about drift and volatility of the risky assets is parameterized by a nonempty set with the form
\be\label{admissable set1}
 {\cal B}=\left\{(\mu_s,\Sigma_s)_{s\geq 0}:(\mu,\Sigma)\;\mbox{are}\,\mathbb{F}\text{-progressively measurable},\,
 \text{and}\ (\mu_s,\Sigma_s\Sigma^{\rm T}_s)\in\mathbb{B},\ \mathbb{P}\otimes {\rm ds}\text{-a.e.} \right\},
\ee
where {$\mathbb{B}$ is a convex and compact subset of $\mathbb{R}^n\times {\cal S}^n_+$, with ${\cal S}^n_+$ being the set of $n\times n$ positive semi-definite real symmetric matrixes. We also assume that $\mathbb{B}$ contains at least one element $(\mu,\Sigma)$ such that $\Sigma\Sigma^{\rm T}$ is positive definite.} The volume of the set $\mathbb{B}$ indicates the amount of uncertainty, the larger the volume, the larger the set of
alternative models, which means there is more uncertainty about the market parameters for the investor.  The investor can improve the accuracy of $\mu$ and $\Sigma$, and lessen $\mathbb{B}$ through buying relevant data asset and obtaining more information about the market parameters.

\subsection{Investment and consumption behavior hypothesis}

Let $c$ be the consumption rate proportional to the investor's wealth, $\pi:=(\pi_1,\pi_2,\cdots,\pi_n)^{\rm T}$ be the proportion of his wealth invested in risky assets. According to the self-financing condition and the formulas~\eqref{riskyasset} and~\eqref{riskfreeasset}, we can obtain the wealth process $X^{t,x; \pi, c, \mu, \sigma}$ with initial time and wealth $(t,x)$, which satisfies the following SDE,
\be \nonumber
 X^{t,x;\pi,c,\mu,\Sigma}_s&=&x+\int_t^s\Big[\,\mu_u^{\rm T}\pi_u
 +r(1-{\bf 1}_n^{\rm T}\pi_u)-(R-r)(1-{\bf 1}_n^{\rm T}\pi_u)^--c_u\,\Big]\,{X^{t,x;\pi,c,\mu,\Sigma}_u}\,{\rm d}u
 \\[2mm]\label{wealth equation}
 &&+\,\int_t^s{X^{t,x;\pi,c,\mu,\Sigma}_u}\pi^{\rm T}_u\Sigma_u \,{\rm d}W_u,\quad
 s\in[t,T],\;x>0,
\ee
where ${\bf 1}_n$ denotes the $n$-dimensional column vector of $1$'s.

The investor will select his portfolio-consumption strategies from the following admissible set with constraints on both portfolio and consumption:
\be\nonumber
 {\cal A}&=&\left\{\,(\pi_s,c_s)_{s\geq 0}:(\pi,c)\;\mbox{are\,$\mathbb{F}$-progressively measurable,}\,
 (\pi_s,c_s)\in \mathbb{A},\, \mathbb{P}\otimes
 {\rm d}s\text{-a.e.},\right.\\[2mm]\label{admissable set2}
 &&\ \int_0^T\left(|\pi_s|^2+c_s\right){\rm d}s<+\infty,\ \text{and}\ X^{t,x;\pi,c,\mu,\Sigma}\ \text{satisfies the
 condition (H)}\},
\ee
where $\mathbb{A}$ is a convex and closed subset of $\mathbb{R}^{n}\times [0,+\infty)$. The integrability condition on $(\pi,c)$ is to guarantee that the wealth process is well defined, while the condition (H) imposed on the wealth process $X^{t,x;\pi,c,\mu,\Sigma}$ depends on the utility maximization problem that we want to solve, and will be specified in (\ref{ClassD}) in the next section.

One typical example of the constraint set is
$\mathbb{A}=\bigotimes_{i=1}^n[\,\underline{\pi}_i,\overline{\pi}_i\,]\times[\,\underline{c},\overline{c}\,]$, where
{$\underline{\pi}_i,\overline{\pi}_i,\underline{c},\overline{c}$ are constants satisfying} $-\infty\leq \underline{\pi}_i\leq0,\;1\leq
\overline{\pi}_i\leq+\infty,\;0\leq\underline{c}\leq\overline{c}\leq+\infty$ for $i=1,\cdots,n$. The portfolio constraint cube
$\bigotimes_{i=1}^n\left[\,\underline{\pi}_i,\overline{\pi}_i\,\right]$ has the following financial interpretations: $\left(\sum_{i=1}^n\overline{\pi}_i-1\right)$ represents the maximum proportion of wealth that the investor is allowed to borrow to invest in the risky assets; $\left(-\sum_{i=1}^n\underline{\pi}_i\right)$ represents the largest short position that the investor is allowed to take;
$\underline{\pi}_i=0$ means prohibition of short selling the $i$th risky asset; $\overline{\pi}_i=1$ means prohibition of borrowing to invest
in the $i$th risky asset; and $-\underline{\pi}_i=\overline{\pi}_i=+\infty$ means no constraints on the $i$th risky asset. Moreover, the consumption
constraint $[\underline{c},\overline{c}]$ means that the investor should keep a minimal consumption proportion $\underline{c}$ for
subsistence purpose, and at the same time, his consumption is also controlled by an upper proportion bound $\overline{c}$ for the sake of future
consumption and investment.

\subsection{The robust utility maximization problem}

The expected utility of the investor derives from the intertemporal consumption and his terminal wealth, which is defined as:
\be\label{object functional}
\mathcal{J}_i(t,x; \pi, c, \mu, \Sigma) :=
\mathbb{E} \left[ \int_t^T \lambda e^{-\rho (s-t)}U_i\left(\,c_sX_s^{t,x; \pi,
c, \mu, \Sigma}\,\right)\mathrm{d}s+ e^{-\rho (T-t)}U_i\left(\,X_T^{t,x; \pi, c, \mu,
\Sigma}\,\right) \bigg|X_t=x \right],
\ee
where $i=1,2$ represent the power and logarithm
utility functions respectively, i.e., $U_1(x)=\frac{1}{p}x^{p}$ with $p\in(-\infty,0)\cup(0,1)$, and $U_2(x)=\ln x$. Herein $\lambda$ and $\rho$ are nonnegative constants, where the parameter $\lambda$ captures the weight of the intertemporal consumption relative to the bequest at maturity $T$, while $\rho$ is the discount factor.

Based on the known information, the investor can not achieve the exact market parameter $(\mu,\Sigma)$, and only know that $(\mu,\Sigma)\in {\cal B}$. Suppose that the investor adopts maximizing robust expected utility criterion to find an optimal investment-consumption strategy in ambiguous finance market, that is, the investor seeks for the optimal strategy that is least affected by model uncertainty \cite{YLZ}. In anticipation of the worst-case scenario, the investor finds $(\pi^{*}, c^{*})\in \mathcal{A}$ and $(\mu^{*}, \Sigma^{*})\in\mathcal{B}$ such that
\begin{equation}\label{stochastic control problem}
J_i(t,x;\mathbb{B}):= \sup_{(\pi, c)\in \mathcal{A}} \inf_{(\mu, \Sigma)\in \mathcal{B}} \mathcal{J}_i(t,x; \pi, c, \mu, \Sigma) = \mathcal{J}_i(t,x; \pi^{*}, c^{*}, \mu^{*}, \Sigma^{*}),\;\;(t,x)\in[0,T]\times(0,+\infty),
\end{equation}
$i=1,2$, where the function $J_i(.;\mathbb{B})$ is called a value function of the stochastic control problem \eqref{stochastic control problem}, representing the optimal robust expected utility.

To close this section, we further specify the condition (H) in the
admissible set $\mathcal{A}$ associated with the maxmin problem
\eqref{object functional}:
\begin{align}\label{ClassD}
\text{ Condition (H)}:=&\left\{\mathbb{E}\left[\int_t^T
U_i(c_sX_s^{t,x;\pi,c,\mu,\Sigma}){\rm d}s\right]<+\infty;\ \text{and the
family}\ U_i\left(X^{t,x;\pi,c,\mu,\Sigma}_{\tau}\right),\
\right.\notag\\[2mm]
&\left. \text{for}\ \tau\in[0,T]\ \text{as an}\
\mathbb{F}\text{-stopping time, is uniformly integrable}\right\},\;\;
i=1,2.
\end{align}
The integrability condition imposed on $U_i\left(X^{t,x;\pi,c,\mu,\Sigma}\right) $ is to include unbounded portfolio and consumption strategies.

\subsection{Pricing model for data asset}

We would like to apply the indifference pricing to price the data asset in a simple framework only considering the information value of the data asset from the buyer's perspective, hence the price in definition \eqref{indifference equality}  is the maximum cost the buyer is willing to pay for the data asset. We derive a concise result that advances the understanding of data asset pricing fundamentals, thereby establishing a foundation for future research on equilibrium pricing of data assets.

The maximum robust expected utility depends on the ambiguity set $\mathbb{B}$, which is determined by the investor's information. We assume that the ambiguity set was $\mathbb{B}_1$ before the investor bought the data asset,  when the investor pays a cost $P_i$ for the data asset, he will obtain a more accurate ambiguity set $\mathbb{B}_2$. According to the indifference pricing idea, we have the following equality
\be\label{indifference equality}
  J_i(t,x;\mathbb{B}_1)
  =J_i(t,x-P_i(t,x;\mathbb{B}_1,\mathbb{B}_2);\mathbb{B}_2),\;\;
  (t,x)\in[0,T]\times(0,+\infty),\;i=1,2,
\ee
where $\mathbb{B}_1\supset\mathbb{B}_2\neq\O$. This is because the investor acquires more information from the data asset, and estimates high-precision parameters, and reduces the volume of the ambiguity set.

The price $P_i$ of the data asset is determined by the equation \eqref{indifference equality}, where the function $J_i(\cdot)$ is   the value function of the stochastic control problem \eqref{stochastic control problem}. In the stochastic control problem, the state equation is SDE \eqref{wealth equation}, and the admissible sets ${\cal A},{\cal B}$ are defined by \eqref{admissable set2} and \eqref{admissable set1}, respectively, and the objective functional ${\cal J}_i$ is defined by \eqref{object functional}.

In the next section, we will show the rationality of the pricing model for data asset, i.e., there exists a unique $P_i$ satisfying \eqref{indifference equality}, moreover, $P_i$ is non-negative.

\section{The explicit form  and some properties of price $P_i$}

In Subsection 3.1, we first utilize the results in \cite{YLZ} to simplify the stochastic control problem \eqref{stochastic control problem} into a saddle point problem of multivariate function with constraints, then we will derive general expressions \eqref{eq-Pi} for the data asset $P_i$. In Subsection 3.2, we will analyze the properties of $P_i$ under the independent investment and consumption constraints to show how the consumption complicates the analyze of data price  $P_i$.

\subsection{The existence and uniqueness of $P_i$}

According to Lemma 3.1,  Theorem 3.2 and Theorem 3.3 in \cite{YLZ}, the stochastic control problem \eqref{stochastic control problem} can be reduced to a saddle point problem of the following multivariate function \eqref{definition of F} with constraints
\be\label{definition of F}
\left\{
\begin{array}{ll}
 F_1(x_q;x_\pi,x_c;x_\mu,x_\Sigma):={p-1\over 2}\,x_\pi^{\rm T}x_\Sigma x_\pi
 +\Big[\,x_\mu^{\rm T}x_\pi+r(1-{\bf 1}_n^{\rm T}x_\pi)^+-R(1-{\bf 1}_n^{\rm T}x_\pi)^-\,\Big]+{{\lambda\over
 p}e^{-x_q}x_c^{p}-x_c},
 \vspace{2mm}\\
 F_2(x_q;x_\pi,x_c;x_\mu,x_\Sigma):=-{1\over 2}\,x_\pi^{\rm T}x_\Sigma x_\pi
 +\Big[\,x_\mu^{\rm T}x_\pi+r(1-{\bf 1}_n^{\rm T}x_\pi)^+-R(1-{\bf 1}_n^{\rm T}x_\pi)^-\,\Big]+\lambda e^{-x_q}\ln x_c- x_c,
\end{array}
\right.
\ee
where $x_q\in \mathbb{R},\,(x_\pi,x_c)\in \mathbb{A},\,(x_\mu,x_\Sigma)\in\mathbb{B}$.

\begin{lemma}\label{lemma3.1} {\bf(Lemma 3.1,  Theorem 3.2 and Theorem 3.3 in \cite{YLZ})} For $i=1,2$, the function $F_i(x_q;\cdot,\cdot;\cdot,\cdot)$ admits at least one saddle point
$(x^*_\pi(x_q),x^*_c(x_q);x^*_\mu(x_q),x^*_\Sigma(x_q))$,
i.e., for any $x_q\in \mathbb{R}$, $(x_\pi, x_c)\in \mathbb{A}$ and $(x_\mu,x_\Sigma)\in \mathbb{B}$, it holds
\be\nonumber
 F_i(x_q;x^*_\pi(x_q),x^*_c(x_q);x_\mu,x_\Sigma)
 &\geq& F_i(x_q;x^*_\pi(x_q),x^*_c(x_q);x^*_\mu(x_q),x^*_\Sigma(x_q))
 \\[2mm]\label{saddle point of F}
 &\geq&
 F_i(x_q;x_\pi,x_c;x^*_\mu(x_q),x^*_\Sigma(x_q)),\;\;i=1,2.
\ee

\noindent Let
\footnote{Note that the saddle point $(x_\pi^*(x_q),x_c^*(x_q);x_\mu^*(x_q),x_\Sigma^*(x_q))$ depends on the ambiguity set $\mathbb{B}$. From \eqref{saddle point value}, we know that the function $G_i$ is unique and independent of the selection of saddle points even though there are multiple saddle points.}
\be\label{defintion of G}
G_i(x_q;\mathbb{B}):=F_i(x_q;x_\pi^*(x_q),x_c^*(x_q);x_\mu^*(x_q),x_\Sigma^*(x_q)),\;\;
i=1,2,
\ee
then the value function of stochastic control problem~\eqref{stochastic control problem} takes the following form of
\be\label{eq-Ji}
 J_i(t,x;\mathbb{B})=
 \left\{
 \begin{array}{ll}
 \frac{x^p}{p}e^{g_1(t;\mathbb{B})},&i=1,
 \vspace{2mm}\\
 g_{21}(t)\ln x+g_{22}(t;\mathbb{B}),\qquad &i=2,
 \end{array}
 \right.
\ee
where
\be\label{eq-g1}
 &&g_1(t;\mathbb{B})=\int_t^T\Big[\,pG_1(g_1(s;\mathbb{B});\mathbb{B})-\rho\Big]{\rm d}s,
 \\[2mm]\label{eq-g21}
&& g_{21}(t)=\left[{\lambda\over\rho}+\left(1-{\lambda\over\rho}\right)
 e^{-\rho(T-t)}\right]{\rm\bf I}_{\{\rho>0\}}
 +\left[1+\lambda(T-t)\right]{\rm\bf I}_{\{\rho=0\}}>0,
 \\[2mm]\label{eq-g22}
&& g_{22}(t;\mathbb{B})=\int_t^T e^{-\rho(s-t)}g_{21}(s)G_2\left( g_{21}(s);\mathbb{B}\,\right)\mathrm{d}s.
\ee
\end{lemma}

\begin{remark}\label{re-saddle point}
From \eqref{saddle point of F}, we know that
\bee
 \inf\limits_{(x_\mu,x_\Sigma)\in\mathbb{B}}
 \sup\limits_{(x_\pi,x_c)\in\mathbb{A}}
 F_i(x_q;x_\pi,x_c;x_\mu,x_\Sigma)
 &\leq&
 \sup\limits_{(x_\pi,x_c)\in\mathbb{A}}
 F_i(x_q;x_\pi,x_c;x^*_\mu(x_q),x^*_\Sigma(x_q))
 \\[2mm]
 &\leq&
 F_i(x_q;x^*_\pi(x_q),x^*_c(x_q);x^*_\mu(x_q),x^*_\Sigma(x_q))
 \\[2mm]
 &\leq&
 \inf\limits_{(x_\mu,x_\Sigma)\in\mathbb{B}}
 F_i(x_q;x^*_\pi(x_q),x^*_c(x_q);x_\mu,x_\Sigma)\\[2mm]
 &\leq& \sup\limits_{(x_\pi,x_c)\in\mathbb{A}}
 \inf\limits_{(x_\mu,x_\Sigma)\in\mathbb{B}}
 F_i(x_q;x_\pi,x_c;x_\mu,x_\Sigma)
\eee
for any $(x_\mu,x_\Sigma)\in\mathbb{B},(x_\pi,x_c)\in\mathbb{A},x_q\in\mathbb{R}$, and $i=1,2$. On the other hand, it is clear that
$$
 \inf\limits_{(x_\mu,x_\Sigma)\in\mathbb{B}}
 \sup\limits_{(x_\pi,x_c)\in\mathbb{A}}
 F_i(x_q;x_\pi,x_c;x_\mu,x_\Sigma)\geq
 \sup\limits_{(x_\pi,x_c)\in\mathbb{A}}
 \inf\limits_{(x_\mu,x_\Sigma)\in\mathbb{B}}
 F_i(x_q;x_\pi,x_c;x_\mu,x_\Sigma)
$$
for any $x_q\in\mathbb{R}$ and $i=1,2$. So we have that
\be\nonumber
 \inf\limits_{(x_\mu,x_\Sigma)\in\mathbb{B}}
 \sup\limits_{(x_\pi,x_c)\in\mathbb{A}}
 F_i(x_q;x_\pi,x_c;x_\mu,x_\Sigma)
 &=&F_i(x_q;x^*_\pi(x_q),x^*_c(x_q);x^*_\mu(x_q),x^*_\Sigma(x_q))
 \\[2mm]\label{saddle point value}
 &=&\sup\limits_{(x_\pi,x_c)\in\mathbb{A}}
 \inf\limits_{(x_\mu,x_\Sigma)\in\mathbb{B}}
 F_i(x_q;x_\pi,x_c;x_\mu,x_\Sigma)
\ee
for any $x_q\in\mathbb{R}$ and $i=1,2$.
\end{remark}

According to Lemma 3.1, we can have the form of the value function $J_i$ if we find the saddle point of the function $F_i$. Next, we will show the rationality of the data asset's pricing model.

\begin{theorem}\label{expression of the price}
For any $(t,x)\in[0,T]\times(0,+\infty), \mathbb{B}_1\supset\mathbb{B}_2\neq\O$,  there exists a unique $P_i(t,x;\mathbb{B}_1,\mathbb{B}_2)\in[0,x),\;i=1,2$.
Moreover, $P_i$ is non-decreasing with $x$ and $\mathbb{B}_1$, and non-increasing with respect to $\mathbb{B}_2$, and takes the following form of
\be\label{eq-Pi}
 P_i(t,x;\mathbb{B}_1,\mathbb{B}_2)=
 \left\{
 \begin{array}{ll}
 x\left[1-
 \exp\left(\frac{g_1(t;\mathbb{B}_1)-g_1(t;\mathbb{B}_2)}
 {p}\right)\right],&i=1,\\[2mm]
 x\left[1-
 \exp\left(\frac{g_{22}(t;\mathbb{B}_1)-g_{22}(t;\mathbb{B}_2)}
 {g_{21}(t)}\right)\right],&i=2,
 \end{array}
 \right.
\ee
where the functions $g_1(t;\mathbb{B}),g_{21}(t),g_{22}(t;\mathbb{B})$ are defined in \eqref{eq-g1}, \eqref{eq-g21} and \eqref{eq-g22}, respectively.
\end{theorem}

\begin{proof}
It is clear that the functions $F_i,G_i,g_1,g_{21}$ and $g_{22}$ are independent of the initial wealth $x>0$, with $g_{21}(t)>0$. Then \eqref{eq-Ji} shows that $J_i$ is continuous and strictly increasing with the initial wealth $x>0$ for any $t\in[0,T]$ and $\mathbb{B}$. For any $y\in[0,x)$, define the function
$$
 \widehat{J}_i(t,y;\mathbb{B}_1,\mathbb{B}_2):
 =J_i(t,x;\mathbb{B}_1)-J_i(t,x-y;\mathbb{B}_2),
$$
the continuity and monotonicity of $J_i$ w.r.t. $x$ lead to the result that $ \widehat{J}_i(y;\mathbb{B}_1,\mathbb{B}_2)$
is continuous and  strictly increasing with respect to $y$.

Moreover, from the definition of $J_i$ in \eqref{stochastic control problem}, we know that
$$
 J_i(t,x;\mathbb{B}_1)
 \leq J_i(t,x;\mathbb{B}_2),\;\;\;(t,x)\in[0,T]\times(0,+\infty),\;\;
 \mathbb{B}_1\supset\mathbb{B}_2\neq\O,\;i=1,2,
$$
and
\bee
 \widehat{J}_i(t,0;\mathbb{B}_1,\mathbb{B}_2)
 =J_i(t,x;\mathbb{B}_1)-J_i(t,x;\mathbb{B}_2)\leq0,,\quad i=1,2,\\
  \lim\limits_{y\rightarrow x^-}\widehat{J}_i(t,y;\mathbb{B}_1,\mathbb{B}_2)
 =J_i(t,x;\mathbb{B}_1)-J_i(t,0;\mathbb{B}_2)>0,\quad i=1,2.
\eee
Thus there exists a unique $P_i\in[0,x)$ satisfying \eqref{indifference equality}.

From \eqref{indifference equality} and \eqref{eq-Ji}, it is not difficult to deduce that $P_i$ takes the form of \eqref{eq-Pi}.
Since $P_i\geq0$ and $x>0$, it is easy to show that $P_i$ is non-decreasing w.r.t. $x$ from \eqref{eq-Pi}.

Assume $\mathbb{B}_2\subset\mathbb{B}_1^1\subset\mathbb{B}^2_1$, then the definition of $J_i$ in \eqref{stochastic control problem} implies that
\bee
 &&J_i(t,x;\mathbb{B}^1_1)\geq J_i(t,x;\mathbb{B}^2_1),\\
 && \widehat{J}_i(t,P_i(t,x;\mathbb{B}^2_1,\mathbb{B}_2);\mathbb{B}^1_1,\mathbb{B}_2) \geq \widehat{J}_i(t,P_i(t,x;\mathbb{B}^2_1,\mathbb{B}_2);\mathbb{B}^2_1,\mathbb{B}_2)
 =0=\widehat{J}_i(t,P_i(t,x;\mathbb{B}^1_1,\mathbb{B}_2);\mathbb{B}^1_1,\mathbb{B}_2).
\eee
Combining the fact that $\widehat{J}_i(t,y;\mathbb{B}^1_1,\mathbb{B}_2)$ is continuous and strictly increasing with respect to $y$, we know that
$$
 P_i(t,x;\mathbb{B}^2_1,\mathbb{B}_2)\geq P_i(t,x;\mathbb{B}^1_1,\mathbb{B}_2),
$$
which means that $P_i$ is non-decreasing with $\mathbb{B}_1$. Repeating the similar argument, we can deduce that $P_i$ is  non-increasing with respect to $\mathbb{B}_2$.
\end{proof}

\begin{remark}
In fact, in our pricing model, the price $P_i$ for the data asset is the maximum purchase price that the investor accepts, if the data asset is valuable for the investor, then it implies that the more initial value the investor has, there will be a greater benefits in the optimal investment and consumption strategies. Hence we have the monotonicity of $P_i$ w.r.t. the initial value $x$, on the other hand, from the point of the seller, the potential customers should have enough money to afford the cost of the data asset, otherwise there is no need to customize the data asset for the customer.
\end{remark}

In the next subsection, we will study the impact of the investment strategies and consumption strategies on the data asset price under the assumption that the admissible investment-consumption strategies set can be separated.

\subsection{Results under independent investment constraint and consumption constraints}

In this subsection, we assume that the admissible investment strategy set is independent of consumption strategy. Concretely speaking, we assume that
\begin{assumption}\label{independent} The admissible investment-consumption strategy set $\mathbb{A}=\mathbb{A}^\pi\times\mathbb{A}^c$, where $\mathbb{A}^\pi$ and $\mathbb{A}^c$ are convex and closed subsets of $\mathbb{R}^n$ and $[0,+\infty)$, respectively.
\end{assumption}
In this case,  we rewrite \eqref{definition of F} as $F_i=F^1+F^2_i$ with
\be\label{eq-F1}
 &&F^1(x_\pi;x_\mu,x_\Sigma):={p-1\over 2}\,x_\pi^{\rm T}x_\Sigma x_\pi
 +\Big[\,x_\mu^{\rm T}x_\pi+r(1-{\bf 1}_n^{\rm T}x_\pi)^+-R(1-{\bf 1}_n^{\rm T}x_\pi)^-\,\Big],
 \\[2mm]\label{eq-F2}
 &&F^2_i(x_q;x_c):=
 \left\{
 \begin{array}{ll}
 {\lambda\over
 p}e^{-x_q}x_c^{p}-x_c,&i=1,
 \vspace{2mm}\\
 \lambda e^{-x_q}\ln x_c- x_c,&i=2.
 \end{array}
 \right.
\ee
Herein, with a slight abuse of notation, we take $p=0$ in the function $F^1$ for $i=2$. Then from \eqref{saddle point value} and \eqref{defintion of G}, we know that

\be\label{de-Gi2}
 G_i(x_q;\mathbb{B})
 =\sup\limits_{(x_\pi,x_c)\in\mathbb{A}}
 \inf\limits_{(x_\mu,x_\Sigma)\in\mathbb{B}}
 F_i(x_q;x_\pi,x_c;x_\mu,x_\Sigma)
 =K(\mathbb{B})
 +f_i(x_q),
\ee
where
\be\label{eq-K,eq-f}
 K(\mathbb{B}):=\sup\limits_{x_\pi\in\mathbb{A}^\pi}
 \inf\limits_{(x_\mu,x_\Sigma)\in\mathbb{B}}
 F^1(x_\pi;x_\mu,x_\Sigma),\qquad
 f_i(x_q):=\max\limits_{x_c\in\mathbb{A}^c}F^2_i(x_q;x_c).
\ee
From Lemma \ref{lemma3.1}, we know that $K(\mathbb{B})$ is a constant that depends on the ambiguity set $\mathbb{B}$.
Under the Assumption \ref{independent}, if we consider only the optimal investment problem without consumption, we have the following result.
\begin{theorem}\label{th-la=0}
Under the Assumption \ref{independent} and $\lambda =0$, $P_i$ has the following clear form of
\be\label{eq-Pi2}
 P_i(t,x;\mathbb{B}_1,\mathbb{B}_2)
 =x\left\{1-e^{[K(\mathbb{B}_1)-K(\mathbb{B}_2)](T-t)}\right\},\quad  i=1,2,
\ee
where $K(\mathbb{B})$ is defined in \eqref{eq-K,eq-f}.  Moreover, $P_i$ is strictly decreasing with respect to $t$ if $K(\mathbb{B}_1)\neq K(\mathbb{B}_2)$.
\end{theorem}

\begin{proof} Since $\lambda=0$, then from \eqref{eq-F2}, we know that  $f_i(x_q)=-\min(c:c\in\mathbb{A}^c)$ is a constant, which is temporarily denoted as $-\underline{c}$. So, by \eqref{de-Gi2}, \eqref{eq-g1}, \eqref{eq-g21} and \eqref{eq-g22}, we have the following computation,
\bee
  g_1(t;\mathbb{B})&=&
  \int_t^T[\,pK(\mathbb{B})
 -p\underline{c}-\rho]{\rm d}s
 =[\,pK(\mathbb{B})
 -p\underline{c}-\rho](T-t),
 \\[2mm]
 g_{21}(t)&=& e^{-\rho(T-t)},
 \\[2mm]
 g_{22}(t;\mathbb{B})&=&
  \int_t^T[\,K(\mathbb{B})
 -\underline{c}\,]\,e^{-\rho(s-t)}g_{21}(s){\rm d}s
 =[\,K(\mathbb{B})
 -\underline{c}\,]\,e^{-\rho(T-t)}(T-t).
\eee
By the expression of $P_i$ in \eqref{eq-Pi}, we can obtain \eqref{eq-Pi2}.

From \eqref{eq-K,eq-f}, we know that $K(\mathbb{B})$ is non-increasing. Since $\mathbb{B}_2\subset \mathbb{B}_1$, we know that $K(\mathbb{B}_1)\leq K(\mathbb{B}_2)$ and $K(\mathbb{B}_1)-K(\mathbb{B}_2)\leq0$. Hence, \eqref{eq-Pi2} implies that $P_i$ is strictly decreasing with respect to $t$ if $K(\mathbb{B}_1)\neq K(\mathbb{B}_2)$.
\end{proof}

\begin{remark}
Theorem \ref{th-la=0} shows that in the optimal investment problem, the information of data asset helps the investor better understand  the financial market so that they can enhance the investment return, but with the passage of time, the time available to utilize this information decreases, leading to a lower price of the data asset.
\end{remark}

\begin{remark}\label{rhola0}
From Theorem \ref{th-la=0}, it is obvious that $K(\mathbb{B})$ plays a crucial role in the data asset price. If we revisit Lemma \ref{lemma3.1} and the calculations in the above proof, it becomes clear that when $\rho=\lambda=0$ and $0\in \mathbb A^c$, we find that the maximum robust utility is $\frac 1 p(xe^{K(\mathbb B)(T-t)})^p$ for the power utility, while  $\ln(xe^{K(\mathbb B)(T-t)})$ for the logarithmic utility. These expressions highlight the impact of investment on the maximum robust utility. Therefore, we call $K(\mathbb{B})$ as The Investment Opportunity Index, which reflects the impacts of investment environment factors such as investment constraints, borrowing/lending rates, market parameters, information and the risk aversion coefficient on investment utility.
\end{remark}

In the following, we will present the expressions for the data asset price under logarithm utility and power utility assumption. Moreover, we will show the properties of $P_i$ under different parameters assumption.

\begin{theorem}\label{th-lo1}
Under Assumption \ref{independent}, the price $P_2$ has the following explicit form of
\be\label{eq-Pi3}
 P_2(t,x;\mathbb{B}_1,\mathbb{B}_2)=\left\{
 \begin{array}{ll}
 x\left\{1-\exp\left[(K(\mathbb{B}_1)-K(\mathbb{B}_2))
 \left({1\over\rho}+\frac{(\rho-\lambda)(T-t)-1}
 {\lambda e^{\rho(T-t)}+\left(\rho-\lambda\right)}\right)\right]\right\},&\rho>0,
 \vspace{2mm}\\
 x\left\{1-\exp\left[(K(\mathbb{B}_1)-K(\mathbb{B}_2))(T-t)
 \frac{1+\frac{1}{2}\lambda(T-t)}{1+\lambda(T-t)}\right]\right\},&\rho=0,
 \end{array}
 \right.
\ee
where $K(\mathbb{B})$ is defined in \eqref{eq-K,eq-f}. Moreover, when $\rho\leq\lambda$ or $\lambda=0$, $P_2$ is strictly decreasing with respect to $t$ if $K(\mathbb{B}_1)\neq K(\mathbb{B}_2)$; whereas when $\rho>\lambda>0$ and $K(\mathbb{B}_1)\neq K(\mathbb{B}_2)$, there exists a $T^*\geq0$ such that $P_2$ is strictly decreasing with respect to $t$ in $[T^*,T]$, and strictly increasing with respect to $t$ in $[0,T^*]$.\footnote{$T^*>0$ if $T$ is large enough, Otherwise, $T^*=0$.}
\end{theorem}

\begin{proof}
From \eqref{de-Gi2}, \eqref{eq-g21} and \eqref{eq-g22}, we have the following computation,
$$
 g_{22}(t;\mathbb{B})
 =\int_t^T e^{-\rho(s-t)}g_{21}(s)
 \left[K(\mathbb{B})
 +f_2(g_{21}(s))\right]\mathrm{d}s
 =K(\mathbb{B})g_{23}(t)
 +\int_t^T e^{-\rho(s-t)}g_{21}(s)
 f_2(g_{21}(s))\,\mathrm{d}s,
$$
where
\bee\label{eq-g23}
 g_{23}(t):=\left\{
 \begin{array}{ll}
 {\lambda \over\rho^2}e^{-\rho(T-t)}\left[e^{\rho(T-t)}-1-\rho(T-t)\right]
 +e^{-\rho(T-t)}(T-t)>0,\;&\rho>0,
 \vspace{2mm}\\
 (T-t)\left[1+\frac{1}{2}\lambda(T-t)\right]>0,&\rho=0,
 \end{array}
 \right.\;\;\forall\;t\in[0,T).
\eee
So we have
$$
 g_{22}(t;\mathbb{B}_1)-g_{22}(t;\mathbb{B}_2)
 =[K(\mathbb{B}_1)-K(\mathbb{B}_2)]g_{23}(t).
$$
By  the following computation,
$$
 \frac{g_{23}(t)}{g_{21}(t)}=\left\{
 \begin{array}{ll}
 \frac{{1\over\rho}\left[\lambda e^{\rho(T-t)}+\left(\rho-\lambda\right)\right]-1
 +(\rho-\lambda)(T-t)}
 {\lambda e^{\rho(T-t)}+\left(\rho-\lambda\right)}
 ={1\over\rho}+\frac{(\rho-\lambda)(T-t)-1}
 {\lambda e^{\rho(T-t)}+\left(\rho-\lambda\right)},\;\;&\rho>0,
 \vspace{2mm}\\
 \frac{1+\frac{1}{2}\lambda(T-t)}{1+\lambda(T-t)}(T-t),&\rho=0,
 \end{array}
 \right.\;\;\forall\;t\in[0,T),
$$
together with \eqref{eq-Pi}, we obtain \eqref{eq-Pi3}.

In the case of $\lambda=0$, from Theorem \ref{th-la=0}, we know that $P_2$ is strictly decreasing with respect to $t$ if $K(\mathbb{B}_1)\neq K(\mathbb{B}_2)$. Next, we assume $\lambda>0$ and $K(\mathbb{B}_1)\neq K(\mathbb{B}_2)$. From \eqref{eq-K,eq-f} and $\mathbb{B}_2\subset \mathbb{B}_1$, we know that $K(\mathbb{B}_1)-K(\mathbb{B}_2)<0$. Hence, the monotonicity of $P_2$ with respect to $t$ is the same as that of $g_{23}(t)/g_{21}(t)$.

In the case of $\rho=0$, it is clear that
$$
 \frac{g_{23}(t)}{g_{21}(t)}=\frac{1+\frac{1}{2}\lambda(T-t)}{\frac{1}{T-t}+\lambda},\quad
 \left(1+\frac{1}{2}\lambda(T-t)\right)^\prime<0,\quad \left(\frac{1}{T-t}+\lambda\right)^\prime>0,\quad
 \frac{g_{23}(t)}{T-t},\frac{g_{21}(t)}{T-t}>0
$$
for any $t\in[0,T)$. So, we deduce that $g_{23}(t)/g_{21}(t)$ and $P_2$ are strictly decreasing with respect to $t$ if $K(\mathbb{B}_1)\neq K(\mathbb{B}_2)$.

In the case of $\rho>0$, we calculate that
$$
 \left(\frac{g_{23}(t)}{g_{21}(t)}\right)^\prime=
 \frac{-(\rho-\lambda)[\lambda e^{\rho(T-t)}+\left(\rho-\lambda\right)]
 +\lambda\rho e^{\rho(T-t)}[(\rho-\lambda)(T-t)-1]}
 {[\lambda e^{\rho(T-t)}+\left(\rho-\lambda\right)]^2}
 =\frac{\widetilde{g}_2(t)}
 {[\lambda e^{\rho(T-t)}+\left(\rho-\lambda\right)]^2},
$$
where
$$
 \widetilde{g}_2(t)=
 \lambda\rho e^{\rho(T-t)}(\rho-\lambda)(T-t) -(\rho-\lambda)^2-(2\rho-\lambda)\lambda e^{\rho(T-t)}.
$$
It is not difficult to check that
\bee
 \widetilde{g}\,_2(T)&=&-(\rho-\lambda)^2
 -(2\rho-\lambda)\lambda
 =-\rho^2<0,
 \\[2mm]
 \widetilde{g}\,_2^\prime(t)&=&-\lambda\rho^2e^{\rho(T-t)}(\rho-\lambda)(T-t)
 -\lambda\rho e^{\rho(T-t)}(\rho-\lambda)
 +(2\rho-\lambda)\lambda\rho e^{\rho(T-t)}
 \\[2mm]
 &=&-\lambda\rho^2e^{\rho(T-t)}(\rho-\lambda)(T-t)
 +\lambda\rho^2 e^{\rho(T-t)}
 =\lambda\rho^2 e^{\rho(T-t)}[1-(\rho-\lambda)(T-t)].
\eee

When $0<\rho\leq\lambda$, then we have $\widetilde{g}\,_2^\prime(t)>0$ and
$\widetilde{g}_2(t)<0$ for any $t\in[0,T]$. It leads to $g_{23}(t)/g_{21}(t)$ and $P_2$ are strictly decreasing with respect to $t$.

When $\rho>\lambda>0$, we have $\widetilde{g}\,_2^\prime(t)>0$ when $t\in(T-1/(\rho-\lambda),T]$ and $\widetilde{g}\,_2^\prime(t)<0$ when $t\in[0,T-1/(\rho-\lambda))$. It implies that $\widetilde{g}_2(t)$ is strictly decreasing  in $[0,T-1/(\rho-\lambda)]$ and strictly increasing in $[T-1/(\rho-\lambda),T]$.

Moreover, it is clear that when $T$ is large enough and $\rho>\lambda>0$, we have
$$
 \widetilde{g}_2(0)
 =\lambda\rho e^{\rho T}(\rho-\lambda)T -(\rho-\lambda)^2-(2\rho-\lambda)\lambda e^{\rho T}>0.
$$
So we know that there exists a $T^*<T$ such that $\widetilde{g}_2(t)>0$ if $[0,T^*)$, and $\widetilde{g}_2(t)<0$ if $(T^*,T]$. This leads to the result when $\rho>\lambda>0$ and $K(\mathbb{B}_1)\neq K(\mathbb{B}_2)$, i.e., $P_2$ is strictly decreasing with respect to $t$ in $[T^*,T]$, and strictly increasing with respect to $t$ in $[0,T^*]$.
\end{proof}

Next, we give a result under power utility assumption. In this case, the problem is more complex than that under logarithm utility assumption. In order to simplify the problem, we assume that there is no constraint on consumption , i.e., $\mathbb{A}^c=[\,0,+\infty)$. In fact, we can generalize the result into the case of $\mathbb{A}^c=[\,\underline{c},\overline{c}\,]$ with $0\leq \underline{c}\leq \overline{c}$ and achieve the similar results via the conclusions in \cite{YLZ}.

\begin{theorem}\label{th-po1}
Under the Assumption \ref{independent}, and $\mathbb{A}^c=[\,0,+\infty),\lambda>0$, then the price $P_1$ has the following explicit form of
\be\label{eq-Pi4}
 P_1(t,x;\mathbb{B}_1,\mathbb{B}_2)
 =x\left\{1-\left[\frac{g_{12}(t,\widetilde{K}(\mathbb{B}_1))}
 {g_{12}(t,\widetilde{K}(\mathbb{B}_2))}
 \right]^{\frac{1-p}{p}}\right\}
\ee
with
\be\label{eq-g12}
 g_{12}(t,y)&\!\!\!\!=\!\!\!\!&
 \left\{
 \begin{array}{ll}
 \lambda^{1\over p-1}e^{y (t-T)}
 -\frac{e^{y (t-T)}-1}{y},\quad
 &y\neq 0,\,\lambda>0,
 \vspace{2mm} \\
 \lambda^{1\over p-1}+(T- t),
 &y=0,\,\lambda>0,
 \end{array}
 \right.\qquad
 \widetilde{K}(\mathbb{B})={\rho-pK(\mathbb{B})\over 1-p},
\ee
where $K(\mathbb{B})$ is defined in \eqref{eq-K,eq-f}. Moreover, when $\rho-\max[p K(\mathbb{B}_1),p K(\mathbb{B}_2)]\leq (1-p)\lambda^{1\over 1-p}$ and $K(\mathbb{B}_1)\neq K(\mathbb{B}_2)$, $P_1$ is strictly decreasing with respect to $t$; whereas when $\rho-\max[p K(\mathbb{B}_1),p K(\mathbb{B}_2)]>(1-p)\lambda^{1\over 1-p}$ and $K(\mathbb{B}_1)\neq K(\mathbb{B}_2)$, there exists a $T^*\geq 0$ such that $P_1$ is  strictly  decreasing with respect to $t$ in $[T^*,T]$, and  strictly increasing with respect to $t$ in $[0,T^*]$.
\end{theorem}

\begin{proof}
From \eqref{de-Gi2}, \eqref{eq-K,eq-f} and \eqref{eq-F2}, we know that
$$
 G_1(x_q;\mathbb{B})
 =K(\mathbb{B})+\max\limits_{x_c\in[0,+\infty)}
 \left({\lambda\over p}e^{-x_q}x_c^{p}-x_c\right)
 =K(\mathbb{B})+\frac{1-p}{p}e^{\frac{x_q}{p-1}}
 \lambda^{\frac{1}{1-p}}.
$$
 By \eqref{eq-g1}, we know $g_1$ satisfies the following ordinary differential equation (ODE),
\be\label{ODE-g1}
 g^\prime_1(t;\mathbb{B})=
 -\Big[\,pK(\mathbb{B})+(1-p)\lambda^{\frac{1}{1-p}}
 e^{\frac{g_1(t;\mathbb{B})}{p-1}}-\rho\Big],
 \quad g_1(T;\mathbb{B})=0.
\ee
Solving the above ODE, we deduce that
\be\label{de-g1-2}
 g_1(t;\mathbb{B})&\!\!\!\!=\!\!\!\!&
 \left\{
 \begin{array}{ll}
 \ln\lambda+(1-p)\ln\left\{\,\left[\,\lambda^{1\over p-1}
 -{1-p\over\rho-pK(\mathbb{B})}\,\right]
 e^{{\rho-pK(\mathbb{B})\over 1-p}
 \left(\,t-T\,\right)}+{1-p\over\rho-pK(\mathbb{B})}\,\right\},
 &\rho\neq pK(\mathbb{B}),\,\lambda>0,
 \vspace{2mm} \\
 \ln\lambda+(1-p)\ln\left[\,\lambda^{1\over p-1}+(T- t)\,\right],
 &\rho=pK(\mathbb{B}),\,\lambda>0.
 \end{array}
 \right.
\ee
Combining \eqref{eq-Pi} and the following computation,
$$
 g_1(t;\mathbb{B}_1)-g_1(t;\mathbb{B}_2)
 =(1-p)\ln g_{12}(t,\widetilde{K}(\mathbb{B}_1))
 -(1-p)\ln g_{12}(t,\widetilde{K}(\mathbb{B}_2)),
$$
we know \eqref{eq-Pi4} holds true, where $g_{12}(t,y)$ and $\widetilde{K}(\mathbb{B})$ are defined in \eqref{eq-g12}.

Next, we investigate the monotonicity of price $P_1$ with respect to $t$ by utilizing \eqref{eq-Pi} and analyzing the monotonic behavior of the $pg_1(t;\mathbb{B}_1)-pg_1(t;\mathbb{B}_2)$. For the convenience of the next discussion, we temporarily denote
$$
 K_i=K(\mathbb{B}_i),\quad \widetilde{K}_i=\widetilde{K}(\mathbb{B}_i),\quad
 g_1^i=g_1(\cdot,\mathbb{B}_i),\;\;i=1,2.
$$
From \eqref{eq-K,eq-f} and $\mathbb{B}_2\subset \mathbb{B}_1$, we know that $K_1\leq K_2$. Next, we assume $\lambda>0$ and $K_1<K_2$.

We first prove that $pg_1^1(t)< pg_1^2(t)$ for any $t\in[0,T)$. It is not difficult to check that
$$
 \partial_y g_{12}(t,y)
 =\lambda^{1\over p-1}e^{y (t-T)}(t-T)
 +e^{y (t-T)}\frac{1-e^{-y (t-T)}-y(t-T)}{y^2}<0,\;\forall\;t\in[0,T),
$$
where we have used the fact that $1-e^x+x\leq0$. So from \eqref{eq-Pi4} and \eqref{eq-Pi}, we know that \footnote{In fact, we can deduce $pg_1^1<pg_1^2$ from \eqref{de-g1-2}, too.}
\be\label{com-g1-g2}
 pg_{12}(t,\widetilde{K}(\mathbb{B}_1))<pg_{12}(t,\widetilde{K}(\mathbb{B}_2)),\quad
 P_1(t,x;\mathbb{B}_1,\mathbb{B}_2)>0,\quad
 pg_1^1(t)< pg_1^2(t),\quad \forall\;x>0,t\in[0,T).
\ee

Next, We study the monotonicity of $(p g_1^1-p g_1^2)(t)$. In the first, by \eqref{de-g1-2}, we calculate that
$$
 g_1^\prime(t;\mathbb{B})=\left\{
 \begin{array}{ll}
 \left[\,\rho-pK(\mathbb{B})
 -(1-p)\lambda^{1\over 1-p}\,\right]\frac{\lambda^{1\over p-1}
 e^{{\rho-pK(\mathbb{B})\over 1-p}
 \left(\,t-T\,\right)}}{\,\left[\,\lambda^{1\over p-1}
 -{1-p\over\rho-pK(\mathbb{B})}\,\right]
 e^{{\rho-pK(\mathbb{B})\over 1-p}
 \left(\,t-T\,\right)}+{1-p\over\rho-pK(\mathbb{B})}},\quad &\rho\neq pK(\mathbb{B}),
 \vspace{2mm} \\
 (1-p)\frac{-1}{\lambda^{1\over p-1}+(T-t)}<0,\quad &\rho= pK(\mathbb{B}).
 \end{array}
 \right.
$$
So, we deduce that
\be\label{der-g1}
  g_1^\prime(t;\mathbb{B})\;\;\left\{
 \begin{array}{ll}
 >0,
 &\rho-pK(\mathbb{B})>(1-p)\lambda^{1\over 1-p},
 \vspace{2mm} \\
 \leq 0,
 &\rho-pK(\mathbb{B})\leq(1-p)\lambda^{1\over 1-p},
 \end{array}
 \right.
\ee

Secondly, from \eqref{ODE-g1}, we know that
\be\label{eq-T}
 (pg_1^1-pg_1^2)^\prime=p^2(K_2-K_1)
 -(1-p)p\lambda^{\frac{1}{1-p}}
 \left[e^{\frac{g_1^1}{p-1}}
 -e^{\frac{g_1^2}{p-1}}\right],\;\;
 (pg_1^1-pg_1^2)^\prime(T)=p^2(K_2-K_1)>0.
\ee
Thus we have
\be\label{ODE-dg1}
 \Delta g_1^{\prime}(t)
 =\left(pg_1^1-pg_1^2\right)^{\prime\prime}(t)
 =p\lambda^{\frac{1}{1-p}}
 \left[e^{\frac{g_1^1}{p-1}}(g_1^1)^\prime(t)
 -e^{\frac{g_1^2}{p-1}}(g_1^2)^\prime(t)\right]
 =a(t)\Delta g_1(t)+b(t),\,
 \Delta g_1(T)>0,
\ee
where
\be\label{eq-a-b}
 \Delta g_1:=\left(pg_1^1-pg_1^2\right)^{\prime},\qquad
 a(t)=\lambda^{\frac{1}{1-p}}
 e^{\frac{g_1^i(t)}{p-1}},\qquad
 b(t)=p\lambda^{\frac{1}{1-p}}
 \left[e^{\frac{g_1^1(t)}{p-1}}
 -e^{\frac{g_1^2(t)}{p-1}}\right](g_1^{3-i})^\prime(t)
\ee
with $i=1$ if $\rho-p K_2\leq(1-p)\lambda^{1\over 1-p}$, and $i=2$ if $\rho-p K_2>(1-p)\lambda^{1\over 1-p}$.

In the case of $\rho-\max(p K_1,p K_2)\leq (1-p)\lambda^{1\over 1-p}$, without loss of generalization, we assume that $\rho-p K_2\leq (1-p)\lambda^{1\over 1-p}$,  let $i=1$ in \eqref{eq-a-b}, from \eqref{der-g1} we know that $(g_1^2)^\prime(t)\leq0$, then $b(t)\leq0$ by \eqref{com-g1-g2}. Applying the comparison theory for ODE \eqref{ODE-dg1}, we know that $\Delta g_1(t)>0$ for any $t\in[0,T]$.\footnote{ In fact, we can directly solve ODE \eqref{ODE-dg1}, and obtain
\bee
 \Delta g_1(t)
 &=&\exp\left[-\int_t^T a(s){\rm d}s\right]\Delta g_1(T)
 -\int_t^T b(u)\exp\left[-\int_t^u a(s){\rm d}s\right] {\rm d}u
 \\[2mm]
&=& p^2(K_2-K_1)\exp\left[-\int_t^T a(s){\rm d}s\right]
 -\int_t^T b(u)\exp\left[-\int_t^u a(s){\rm d}s\right] {\rm d}u,
\eee which implies the same result, too.} So we deduce that $p[g_1(t;\mathbb{B}_1)-g_1(t;\mathbb{B}_2)]$ is strictly increasing with respect to $t$, and $P_1$ is strictly decreasing with respect to $t$ in $[0,T]$ via \eqref{eq-Pi}.

In the case of $\rho-\max(p K_1,p K_2)>(1-p)\lambda^{1\over 1-p}$,  let $i=2$ in \eqref{eq-a-b}, from \eqref{der-g1} and \eqref{com-g1-g2}, we know that $(g_1^1)^\prime>0$ and $b(t)>0$ for any $t\in[0,T)$. Since \eqref{eq-T} and
\bee
 \Delta g_1(0)&=&
 \frac{p\left[\,(\rho-pK_1)\lambda^{1\over p-1}-(1-p)\,\right]
 e^{-{\rho-pK_1\over 1-p}T}}{\,\left[\,\lambda^{1\over p-1}
 -{1-p\over\rho-pK_1}\,\right]
 e^{-{\rho-pK_1\over 1-p}T}+{1-p\over\rho-pK_1}}
 -\frac{p\left[\,(\rho-pK_2)\lambda^{1\over p-1}-(1-p)\,\right]
 e^{-{\rho-pK_2\over 1-p}T}}{\,\left[\,\lambda^{1\over p-1}
 -{1-p\over\rho-pK_2}\,\right]
 e^{-{\rho-pK_2\over 1-p}T}+{1-p\over\rho-pK_2}}
 \\[2mm]
 &=&p e^{{pK_1-\rho\over 1-p}T}\left\{\frac{\left[\,(\rho-pK_1)\lambda^{1\over p-1}-(1-p)\,\right]}
 {{1-p\over\rho-pK_1}+o(1)}
 -\frac{\left[\,(\rho-pK_2)\lambda^{1\over p-1}-(1-p)\,\right]
 }{{1-p\over\rho-pK_2}+o(1)}e^{{p(K_2-K_1)\over 1-p}T}\right\}
 \\[2mm]
 &\rightarrow&0^-,
\eee
as $T\rightarrow+\infty$. So, we know that $\Delta g_1(0)<0$ provided $T$ large enough. Next, we assume that $T$ is enough large, and $\Delta g_1(0)<0$.

Define $T^*=\sup\{t\in[0,T]:\Delta g_1(t)\leq 0\}\in(0,T)$. It is clear that $\Delta g_1\in C^\infty[0,T],\Delta g_1(T^*)=0$ and $\Delta g_1(t)>0$ for any $t\in(T^*,T]$. Moreover, applying the comparison theory for ODE \eqref{ODE-dg1} in the interval $[0,T^*]$, we know that $\Delta g_1(t)<0$ for any $t\in[0,T^*)$.\footnote{ In fact, as above, we can directly solve ODE \eqref{ODE-dg1} in the interval $[0,T^*]$ and achieve the expression of $\Delta g_1$ in $[0,T^*]$, which implies that $\Delta g_1(t)<0$ for any $t\in[0,T^*)$, too.} Combining \eqref{eq-Pi}, we deduce that  $P_1$ is  strictly  decreasing with respect to $t$ in $[T^*,T]$, and  strictly increasing with respect to $t$ in $[0,T^*]$.

\end{proof}

\begin{remark}\label{rhola1}
From Theorems \ref{th-lo1} and \ref{th-po1}, we know that taking consumption into account complicates the analysis, so that the price of data asset isn't monotonic w.r.t. the time $t$ any more.  This is because the investor must balance and allocate his resources between investment and immediate consumption. When the utility discount is not significant, investor tends to focus more on the future, thus leaning more towards investment. Conversely, when the utility discount is substantial, investor prioritises immediate consumption. The value of information is crucial for the investment, hence when the utility discount is not significant, the value of information is greater, and thus the price decreases over time t. On the other hand, when the utility discount is substantial, the value of information relatively diminishes, making the issue more complex, the price of the data asset initially increases and then decreases over time $t$. From the perspective of the seller, they could choose different sales timings for different customers.
\end{remark}

In this section, we know $K(\mathbb{B}_i)$ plays an essential role in the data asset price, so we would like to figure out the explicit expressions of $K(\mathbb{B}_i)$ under specific conditions in the next section.

\section{More specific examples}

In this section, we  will look further into more specific examples and present more concrete results. Thanks to Theorem \ref{th-la=0}, Theorem \ref{th-lo1} and  Theorem \ref{th-po1}, we know that when $\lambda=0$, $P_1,P_2$ are defined by \eqref{eq-Pi2}; $P_1$ is defined by \eqref{eq-Pi4}  when $\lambda>0, \underline{c}=0, \overline{C}=+\infty$, and $P_2$ is defined by \eqref{eq-Pi3} when $\lambda>0$. Hence it is sufficient to find the constant $K(\mathbb{B}_i), i=1,2$ to obtain the explicit expressions of the data asset price. Next we will show the specific expressions $K(\mathbb{B}_i), i=1,2$ in various examples.

\begin{assumption}\label{no correlation}
Assume that $n=1$, and $\mathbb{A}=[\,\underline{\pi},\overline{\pi}\,]
\times[\,\underline{c},\overline{c}\,]$, and $\mathbb{B}_i=[\,\underline{\mu}_i,\overline{\mu}_i\,]\times
[\,\underline{\sigma}_i^2,\overline{\sigma}_i^2\,]$, where
$\underline{\pi},
\overline{\pi},\underline{c},\overline{c},\underline{\mu}_i,
\overline{\mu}_i,\underline{\sigma}_i,\overline{\sigma}_i$ are constants
satisfying $-\infty\leq \underline{\pi}\leq 0,1\leq \overline{\pi}\leq+\infty,\,
0\leq \underline{c}\leq\overline{c}\leq +\infty$, and $-\infty<\underline{\mu}_1\leq \underline{\mu}_2\leq \overline{\mu}_2\leq \overline{\mu}_1<+\infty,\,
0<\underline{\sigma}_1\leq\underline{\sigma}_2\leq
\overline{\sigma}_2\leq\overline{\sigma}_1<+\infty,i=1,2$.
\end{assumption}

\begin{theorem}\label{th-ex1}
Under Assumption \ref{no correlation}, the $K(\mathbb{B}_i)$ takes the form in Table \ref{table1}, where
\bee
  \beta_{i,1}{:=}{\underline{\mu}_i-R\over (1-p)\overline{\sigma}^2_i},\qquad\qquad
  \beta_{i,2}{:=}{\underline{\mu}_i-r\over (1-p)\overline{\sigma}^2_i},\qquad\qquad
  \beta_{i,3}{:=}{\overline{\mu}_i-r\over (1-p)\overline{\sigma}^2_i},\hspace{1.7cm}
  \\[2mm]
  K^1_{i,1}:=\frac{p-1}{2}\overline\pi^2\overline\sigma_i^2
  +(\underline\mu_i-R)\overline\pi+R,\qquad
  K^1_{i,2}:={(\underline{\mu}_i-R)^2\over 2(1-p)\overline{\sigma}_i^2}+R,\qquad
  K^1_{i,3}:=\frac{p-1}{2}\overline\sigma_i^2+\underline{\mu}_i,\;\,
  \\[2mm]
  \quad K^1_{i,4}:={(\underline{\mu}_i-r)^2\over 2(1-p)\overline{\sigma}_i^2}+r,\quad
  K^1_{i,5}:=r,\quad
  K^1_{i,6}:={(\overline{\mu}_i-r)^2\over 2(1-p)\overline{\sigma}_i^2}+r,\quad
  K^1_{i,7}:=\frac{p-1}{2}\underline\pi^2\overline\sigma_i^2
  +(\overline\mu_i-r)\underline\pi+r.
\eee

\begin{table}[h!]
    \caption{the constant $K(\mathbb{B}_i),\;\;i=1,2$}
    \ \ \ \ \ \ \ \ \ \ \ \
    \begin{tabular}{c c c c c c c c}
      \hline
      $\!\!\!$& $\!\!\!\!\!\beta_{i,1}\geq \overline{\pi}\!\!$& $\!\!1\leq \beta_{i,1}\leq \overline{\pi}\!\!$ &$\!\!\beta_{i,1}\leq1\leq\beta_{i,2}\!\!$ & $\!\!0\leq\beta_{i,2}\leq1\!\!$
      & $\!\!\beta_{i,2}\leq 0\leq \beta_{i,3}\!\!$& $\!\!\underline{\pi}\leq\beta_{i,3}\leq 0\!\!$& $\!\!\beta_{i,3}\leq \underline{\pi} $\\[1mm]
      \hline
      $\!\!\!K(\mathbb{B}_i)$ &$\;\;K^1_{i,1}$ &$K^1_{i,2}$ & $K^1_{i,3}$ & $K^1_{i,4}$ & $K^1_{i,5}$ & $K^1_{i,6}$& $K^1_{i,7}$\\[1mm]
      \hline
    \end{tabular}  \label{table1}
  \end{table}
\end{theorem}

\begin{proof}From \eqref{eq-K,eq-f} and Theorem 4.2 in \cite{YLZ}, we know that
$$
 K(\mathbb{B}_i)=\sup\limits_{x_\pi\in[\,\underline{\pi},\overline{\pi}\,]}
 \inf\limits_{(x_\mu,x_\Sigma)\in[\,\underline{\mu}_i,\overline{\mu}_i\,]\times
[\,\underline{\sigma}_i^2,\overline{\sigma}_i^2\,]}
 F^1(x_\pi;x_\mu,x_\Sigma)=F^1(x_{\pi,i}^*;x_{\mu,i}^*,x^*_{\Sigma,i}),
$$
where  $x^*_{\Sigma,i}=\overline{\sigma}_i$, and $x_{\pi,i}^*,x_{\mu,i}^*$ takes the form in Table \ref{table2}. So it is not difficult to obtain the conclusions of $K(\mathbb B_i)$ by \eqref{eq-F1}.\end{proof}

\begin{table}[h!]
    \caption{the values of $x_{\pi,i}^*$ and $x_{\mu,i}^*, i=1,2$}
    \ \ \ \ \ \ \ \ \ \ \ \
    \begin{tabular}{c c c c c c c c}
      \hline
      $\!\!\!$& $\!\!\!\!\!\beta_{i,1}\geq \overline{\pi}\!\!$& $\!\!1\leq \beta_{i,1}\leq \overline{\pi}\!\!$ &$\!\!\beta_{i,1}\leq1\leq\beta_{i,2}\!\!$ & $\!\!0\leq\beta_{i,2}\leq1\!\!$
      & $\!\!\beta_{i,2}\leq 0\leq \beta_{i,3}\!\!$& $\!\!\underline{\pi}\leq\beta_{i,3}\leq 0\!\!$& $\!\!\beta_{i,3}\leq \underline{\pi} $\\[1mm]
      \hline
      $\!\!\!x_{\pi,i}^*$ &$\;\;\overline{\pi}$ &$\beta_{i,1}$ & $1$ & $\beta_{i,2}$ & $0$ & $\beta_{i,3}$& $\underline{\pi}$\\[1mm]
      \hline
      $\!\!\!x_{\mu,i}^*$ &$\;\;\underline{\mu}_i$ &$\underline{\mu}_i$ & $\underline{\mu}_i$ & $\underline{\mu}_i$ & $r$ & $\overline{\mu}_i$& $\overline{\mu}_i$\\[1mm]
      \hline
    \end{tabular}  \label{table2}
\end{table}

\begin{remark}
The above Theorem implies that not all data assets hold value, only those of sufficient quality are valuable. For instance, if the estimation based on a particular data asset results in
$\underline{\mu}_2\leq r\leq \overline{\mu}_2$, then we have $\beta_{i,2}\leq 0\leq \beta_{i,3},  K(\mathbb{B}_i)=r, i=1,2$, which means  that the investor would fully invest in a risk-free asset regardless of the presence or absence of this data asset, which means this data asset is devoid of value and its purchase price would be zero.
\end{remark}

In the following, we suppose that  the ambiguities about drift and volatility are correlated, a higher return is associated with a larger risk.

\begin{assumption}\label{correlation}
Assume that $n=1$, and $R=r,\mathbb{A}=\mathbb{R}\times[\,\underline{c},\overline{c}\,]$ and
$\mathbb{B}_i=\{(\mu,\sigma):\mu=\underline{\mu}_i+\alpha,
\sigma=\underline{\sigma}_i^2+k\alpha^q,\alpha\in[\,0,\overline{\alpha}_i\,]\,\}$,
where $\underline{c},\overline{c},\underline{\mu}_i,\underline{\sigma}_i,k,q,\overline{\alpha}_i$ are constants satisfying $\underline{c}\leq \overline{c}$, and $\underline{\mu}_1\leq\underline{\mu}_2$ and
$$
\underline{\sigma}_2\geq\underline{\sigma}_1\geq0,\quad
k>0,\quad 0<q<1,\quad \overline{\alpha}_2\geq0,\quad
\overline{\alpha}_1\geq\max\left[\overline{\alpha}_2
+(\underline{\mu}_2-\underline{\mu}_1),\,
\left(\overline{\alpha}_2^q+\frac{\underline{\sigma}_2^2-
\underline{\sigma}_1^2}{k}\right)^{\frac{1}{q}}\right].
$$
\end{assumption}

 The limiting case $q=1$ means that the relationship
between the ambiguity about drift and the ambiguity about the
volatility square is linear, which is just Example 2.4 in
\cite{Epstein}. The other spectrum $q=0$ means no ambiguity about
volatility. Finally, $0<q<1$ means that the relationship between the ambiguity about drift and the ambiguity about the volatility
square is sub-linear, which is just example 4.3 in \cite{YLZ}.

\begin{theorem}\label{th-ex2}
Under Assumption \ref{correlation}, the constant  $K(\mathbb{B}_i)=(\mu_i^*-r)^2/(2(1-p)(\sigma^{*}_i)^2)+r$, where
$$
 \mu^*_i=\underline{\mu}_i+\alpha^*_i,\;\;
 \sigma^*_i=\sqrt{\underline{\sigma_i}^2+k(\alpha^*_i)^q},\;\;
 \alpha^*_i=\left\{
 \begin{array}{ll}
 r-\underline{\mu}_i,\;\; &-\overline{\alpha}_i<\underline{\mu}_i-r\leq 0;
 \\[2mm]
 \widehat{\alpha}_i, &0<\underline{\mu}_i-r
 <[\,2\underline{\sigma}_i^2\overline{\alpha}_i^{1-q}
 +{k(2-q)\overline{\alpha}_i\,]/(kq)};
 \\[2mm]
 \overline{\alpha}_i,&\mbox{otherwise},
 \end{array}
 \right.
$$
where $\widehat{\alpha}_i$ is the unique solution of the following algebra equations (for the case $\underline{\mu}_i-r>0$),
\bee
 2\underline{\sigma}_i^2+{k(2-q)\alpha^q-k q(\underline{\mu}_i-r)\alpha^{q-1}}=0.
\eee
\end{theorem}

The proof can be refer to Theorem 4.4 in \cite{YLZ}.

\begin{assumption}\label{BP}{\bf(Assumption in \cite{BP} and \cite{PWZ})}
Assume $R=r,\mathbb{A}=\mathbb{R}^n\times[\,\underline{c},\overline{c}\,]$, where $\underline{c},\overline{c}$ are constants satisfying $\underline{c}\leq \overline{c}$, and
$\mathbb{B}_i=\{(\mu,\Sigma):(\mu-\widehat{\mu})^{\rm T}\Sigma^{-1}(\mu-\widehat{\mu})\leq \epsilon_i^2, \Sigma\in\mathbb{B}^\Sigma_i\}$,
where $\widehat{\mu}$ is a constant vector, $\epsilon_i$ are  positive constants satisfying $\epsilon_1\geq \epsilon_2$, and $\mathbb{B}^\Sigma_i$ are non-empty compact subset of ${\cal S}^n_+$ satisfying that there exist constants $\delta>0$ such that
$$
 \xi^{\rm T}\Sigma\xi\geq\delta|\xi|^2,\;\;\forall\;\;
 \xi\in\mathbb{R}^n,\Sigma\in\mathbb{B}^\Sigma_1,\qquad
 \mathbb{B}^\Sigma_1\supset\mathbb{B}^\Sigma_2.
$$
\end{assumption}

\begin{theorem}\label{th-BP}
Under Assumption \ref{BP}, the constant
\be\label{eq-H}
 K(\mathbb{B}_i)=
 \frac{\left[\max\left(\min\limits_{x_\Sigma\in \mathbb{B}^\Sigma_i}H(x_\Sigma)-\epsilon_i,0\right)\right]^2}
 {2(1-p)}+r,\qquad
 H(x_\Sigma):=\sqrt{\left(\widehat{\mu}-r{\bf 1}_n\right)^{\rm T}
 (x_\Sigma)^{-1}\left(\widehat{\mu}-r{\bf 1}_n\right)}.
\ee
\end{theorem}

The proof can be found in \cite{BP}. For the convenience of reading, we provide a simplified proof based on our results.

\begin{proof}
From \eqref{eq-K,eq-f} and repeating the similar argument in Remark \ref{re-saddle point}, we know that
$$
 K(\mathbb{B}_i)=\min\limits_{x_\Sigma\in \mathbb{B}^\Sigma_i}\widehat{K}_1(x_\Sigma),\;\;
 \widehat{K}_1(x_\Sigma):=
 \max\limits_{x_\pi\in \mathbb{R}^n}\widehat{K}_2(x_\Sigma,x_\pi),\;\;
 \widehat{K}_2(x_\Sigma,x_\pi):=
 \min\limits_{x_\mu\in\{x_\mu:(x_\mu,x_\Sigma)\in \mathbb{B}_i\}}F^1(x_\pi;x_\mu,x_\Sigma).
$$

In the first, we calculate $\widehat{K}_2(x_\Sigma,x_\pi)$. According to Karush-Kuhn-Tucker Condition, we know that $x_\mu^*$ and the Lagrange multiplier $\Lambda^*$ satisfy $\Lambda^*\geq0$, and
$$
 x_\pi+2\Lambda^*(x_\Sigma)^{-1}(x_\mu^*-\widehat{\mu})=0,\;\;
 \Lambda^*\left[(x_\mu^*-\widehat{\mu})^{\rm T}(x_\Sigma)^{-1}(x_\mu^*-\widehat{\mu})-\epsilon_i^2\right]=0,\;\; (x_\mu^*-\widehat{\mu})^{\rm T}(x_\Sigma)^{-1}(x_\mu^*-\widehat{\mu})-\epsilon_i^2\leq0.
$$
So we deduce that
$$
  x_\pi^{\rm T}x_\Sigma x_\pi=\left(2\Lambda^*(x_\Sigma)^{-1}(x_\mu^*-\widehat{\mu})\right)^{\rm T}x_\Sigma \left(2\Lambda^*(x_\Sigma)^{-1}(x_\mu^*-\widehat{\mu})\right)=4(\Lambda^*)^2(x_\mu^*-\widehat{\mu})^{\rm T}(x_\Sigma)^{-1}(x_\mu^*-\widehat{\mu})=4(\epsilon_i\Lambda^*)^2,
$$
and
$$
 \Lambda^*= \frac{\sqrt{x_\pi^{\rm T}x_\Sigma x_\pi}}{2\epsilon_i},\;\;
 x_\mu^*=\widehat{\mu}- \frac{\epsilon_i x_\Sigma x_\pi}{\sqrt{x_\pi^{\rm T}x_\Sigma x_\pi}},\;\;
 \widehat{K}_2(x_\Sigma,x_\pi)={p-1\over 2}\,x_\pi^{\rm T}x_\Sigma x_\pi-\epsilon_i\sqrt{x_\pi^{\rm T}x_\Sigma x_\pi}+\left(\widehat{\mu}-r{\bf 1}_n\right)^{\rm T}x_\pi+r
$$
provided $x_\pi\neq0$. Note that $\widehat{K}_2(x_\Sigma,x_\pi)$ still takes the above form even if $x_\pi=0$.

In order to obtain $\widehat{K}_1(x_\Sigma)$, we denote $z=\sqrt{x_\pi^{\rm T}x_\Sigma x_\pi}$, then we know that
$$
 \widehat{K}_1(x_\Sigma)
 =\max\limits_{z\geq0}\widehat{K}_{3}(z,x_\Sigma),\quad
 \widehat{K}_{3}(z,x_\Sigma):=
 \max\left\{\widehat{K}_2(x_\Sigma,x_\pi):x_\pi^{\rm T}x_\Sigma x_\pi=z^2\right\}.
$$

According to Karush-Kuhn-Tucker Condition, we know that $x_\pi^*$ and the Lagrange multiplier $\Lambda^*$ satisfy $\Lambda^*\leq0$, and
$$
 \widehat{\mu}-r{\bf 1}_n+2\Lambda^* x_\Sigma x_\pi^*=0,\qquad
 (x_\pi^*)^{\rm T}x_\Sigma x_\pi^*=z^2.
$$

As above, we deduce that
$$
  \Lambda^*=\frac{-H(x_\Sigma)}{2z},\qquad
  x_\pi^*=\frac{z(x_\Sigma)^{-1}\left(\widehat{\mu}-r{\bf 1}_n\right)}{H(x_\Sigma)},\qquad
  \widehat{K}_{3}(z,x_\Sigma)={p-1\over 2}\,z^2+[H(x_\Sigma)-\epsilon_i]z+r
$$
provided $\widehat{\mu}-r{\bf 1}_n\neq0$, where $H(x_\Sigma)$ is defined in \eqref{eq-H}. Note that $\widehat{K}_3(z,x_\Sigma)$ still takes the above form even if $\widehat{\mu}-r{\bf 1}_n=0$.

Since
$$
 \partial_z\widehat{K}_3=
 (p-1)\left[z-\frac{H(x_\Sigma)-\epsilon_i}{1-p}\right]\;\;
 \left\{
 \begin{array}{ll}
 >0,&z<[H(x_\Sigma)-\epsilon_i]/(1-p),
 \\[2mm]
 <0,&z>[H(x_\Sigma)-\epsilon_i]/(1-p),
 \end{array}
 \right.
$$
we deduce that
$$
 \widehat{K}_1(x_\Sigma)
 =\frac{\left[\max\left(H(x_\Sigma)-\epsilon_i,0\right)\right]^2}
 {2(1-p)}+r.
$$
From the above expression of $\widehat{K}_1(x_\Sigma)$, it is not difficult to obtain the conclusion.
\end{proof}

\begin{assumption} \label{as-sample sizes} Assume $R=r, n=1,\mathbb{A}=\mathbb{R}\times[\,0,+\infty\,]$, and
$$
 \mathbb{B}_i=\left[\,\widehat{\mu}-{t_{1-\alpha/2}(N_i-1)s\over \sqrt{N_i}},\,\widehat{\mu}+{t_{1-\alpha/2}(N_i-1)s\over \sqrt{N_i}}\,\right]\times\left[\,{(N_i-1)s^2\over \chi^2_{1-\alpha/2}(N_i-1)},\,{(N_i-1)s^2\over \chi^2_{\alpha/2}(N_i-1)}\,\right],
$$
where $\widehat{\mu}$ is the sample mean, and $s^2$ is the sample variance, and $N_i$ are the sample sizes satisfying $N_2>N_1$, and $1-\alpha$ is the confidence level.
\end{assumption}

\begin{theorem}\label{th-sample sizes}
Under Assumption \ref{as-sample sizes}, the constant
$$
 K(B_i)={\max\left[|\widehat{\mu}-r|
 -\frac{t_{1-\alpha/2}(N_i-1)s}{\sqrt{N_i}},\,0\right]^2\over 2(1-p)}
 {\chi^2_{\alpha/2}(N_i-1)\over(N_i-1)s^2}+r.
$$
\end{theorem}

The proof follows from Theorem \ref{th-ex1}.

\section{Numerical Illustrations}
Based on the above theoretical discussions, now we describe how to apply our data asset pricing method based on numerical illustrations. For the sake of simplicity, we use the simplest data asset: a dataset to study the nature of data asset price. The basic idea is as follows. The buyer has an existing dataset $\mathcal{X}$ and plans to purchase a new dataset $\mathcal{Y}$. The dataset seller has knowledge about both $\mathcal{X}$ and $\mathcal{Y}$, and calculates the maximum price $P$ that the buyer is willing to pay for $\mathcal{Y}$. Following \cite{GI}, we quantify the ambiguity set as confidence intervals.

The set $\mathbb{B}$ is derived using the following procedures. First, we get the confidence intervals of $\mu$ and $\sigma$ from $\mathcal{X}$ and $\mathcal{Y}$ as $\mathcal{I}^{j}_{\mu}$ and $\mathcal{I}^{j}_{\sigma}$, where $j\in\{\mathcal{X},\mathcal{Y}\}$. Then,
$$
\mathbb{B}_1:= \left(\mathcal{I}^{\mathcal{X}}_{\mu}\lor \mathcal{I}^{\mathcal{Y}}_{\mu}\right)\times \left(\mathcal{I}^{\mathcal{X}}_{\sigma}\lor \mathcal{I}^{\mathcal{Y}}_{\sigma}\right),\qquad
\mathbb{B}_2:= \mathcal{I}^{\mathcal{Y}}_{\mu}\times \mathcal{I}^{\mathcal{Y}}_{\sigma},
$$
where for two intervals $I_1$ and $I_2$, $I_1\lor I_2$ means the smallest interval containing $I_1$ and $I_2$.
It is straightforward that $\mathbb{B}_1 \supset \mathbb{B}_2$, so that our assumption for \eqref{indifference equality} is satisfied.

For $i\in\{\mathcal{X},\mathcal{Y}\}$, denote $\mathbb{B}_i = [\underline{\mu}^i,\bar{\mu}^i]\times [\underline{\sigma}^i,\bar{\sigma}^i]$ and $\hat{\mu}_i=\frac{\underline{\mu}^i+\bar{\mu}^i}{2}, \hat{t}_i=\frac{\bar{\mu}^i-\underline{\mu}^i}{2}$, we have
$$
	K(\mathbb{B}_i)=\frac{\max\left[|\hat{\mu}_i-r|-\hat{t}_i,0\right]^2}{2(1-p)}\frac{1}{\overline{\sigma}_i}+r,\quad i=1,2.
$$

To illustrate, we assume the return has a normal distribution $\mathcal{N}(\mu,\sigma^2)$ with true values $\mu=0.1$ and $\sigma=0.2$. However, we only observe two sample datasets $\mathcal{X}$ and $\mathcal{Y}$ simulated from this true distribution, with sample sizes $N_1$ and $N_2$, respectively. We fix the size of the small dataset $\mathcal{X}$ as $N_1=1000$, and vary the size of the larger dataset $\mathcal{Y}$ as $N_2=2000,4000,\ldots,20000$. For each pair of simulated dataset $(\mathcal{X},\mathcal{Y})$, we fit a normal distribution for both and define $(\underline{\mu}^i,\bar{\mu}^i)$ as the 95\% confidence interval for the estimation of mean and $\underline{\sigma}^i,\bar{\sigma}^i$ as the 95\% confidence interval for the estimation of the standard deviation, for $i\in\{\mathcal{X},\mathcal{Y}\}$, define $(\mathbb{B}_1,\mathbb{B}_2)$ correspondingly, and obtain the price $P$. We repeat the above procedure for $M=10000$ times, i.e., by generating $M$ pairs of the simulated dataset $\{(\mathcal{X}^j,\mathcal{Y}^j)\}_{j=1}^M$ and obtaining $M$ values of the prices $\{P_j\}_{j=1}^M$.We then report the average of these $M$ prices as the data asset price $P$.

The results are shown in Figure \ref{fig:priceSample}. Consistent with our intuition, the price of the new larger dataset over the smaller dataset increases in its sample size.

\begin{figure}[htbp!]
\centering
	\includegraphics[width=0.7\textwidth]{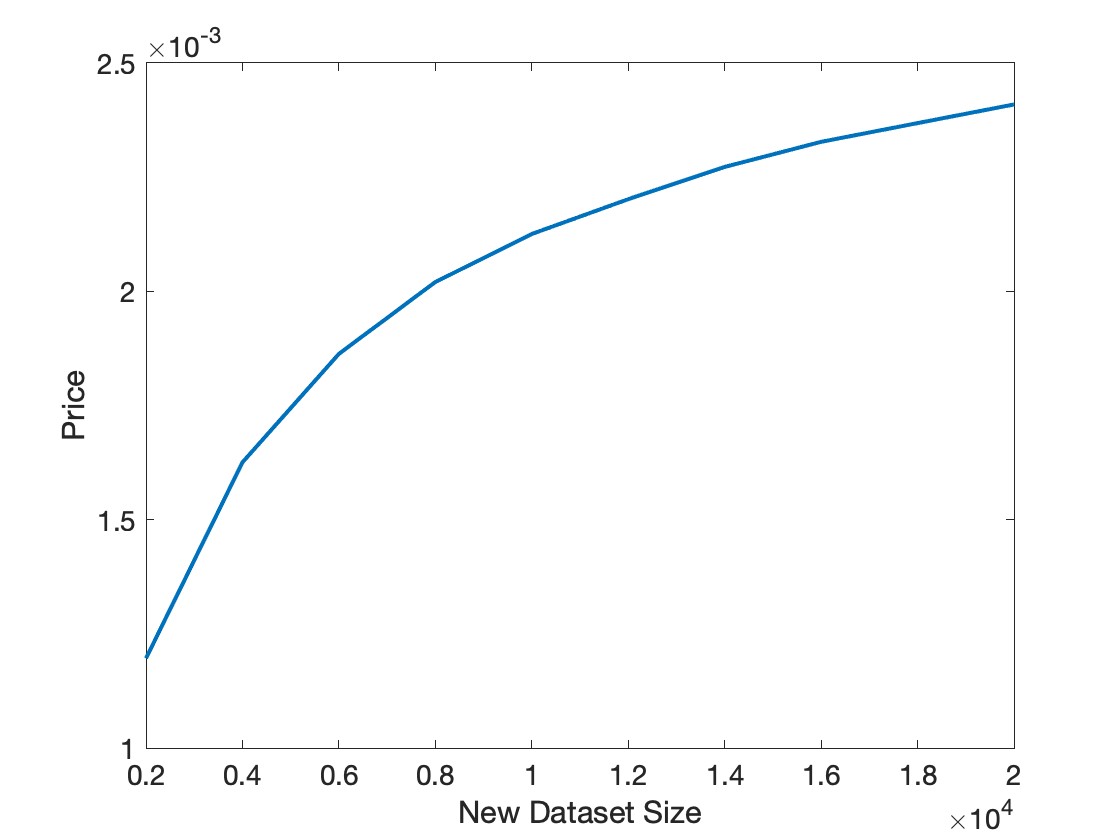}
	\caption{Dataset price $P$ with respect to new data size. Parameters: $x=1$, $\rho = 0.1$, $r = 0.04$, $p = 0.5$, and $\lambda = 0.2$.}\label{fig:priceSample}
\end{figure}

Next, we study the impact of the true values of $\mu$ and $\sigma$ on the price of the new dataset. To this end, we fix $N_1=1000$ and $N_2=5000$, and vary the true values of $\mu\in[0.04,0.3]$ and $\sigma\in[0.1,0.5]$. For each pair of true values $(\mu,\sigma)$, we generate $M=10000$ pairs of simulated datasets $\{(\mathcal{X}^j,\mathcal{Y}^j)\}_{j=1}^M$ and obtain the price $P$ using the same procedure as above. The result in Figure \ref{fig:priceSample2} shows that the price of the new dataset is increasing in $\mu$ and decreasing $\sigma$. In other words, the new dataset is more valuable if the stock has a higher return and lower volatility, or in general, in an overall more profitable and less risky market environment.

\begin{figure}[htbp!]
\centering
  \includegraphics[width=0.7\textwidth]{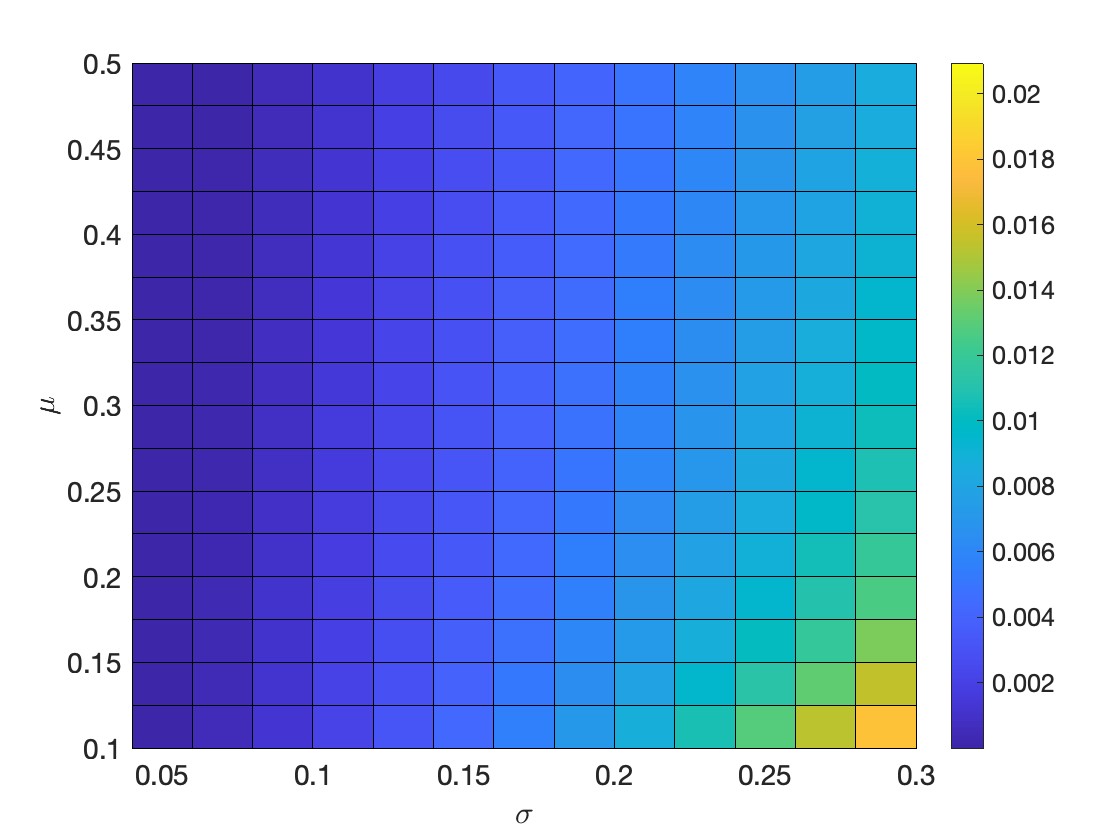}
    \caption{Dataset price $P$ with respect to various $\mu$ and $\sigma$, for $\mu\in[0.04,0.3]$ and $\sigma\in[0.1,0.5]$. Parameters: $x=1$, $\rho = 0.1$, $r = 0.04$, $p = 0.5$, and $\lambda = 0.2$.}\label{fig:priceSample2}
\end{figure}

Finally, we study the impact of the risk aversion $1-p$ and the weight on the intertemporal consumption $\lambda$  on the dataset price, by varying $p\in[0.1,0.9]$ and $\lambda\in[0.1,1]$. The result in Figure \ref{fig:priceSample3} shows that the price is lower for investors with higher risk aversion and lower weight on the intertemporal consumption. Indeed, they tend to allocate more in the risky asset and hence assign a higher value for the price of the data asset.

\begin{figure}[htbp!]
  \centering
    \includegraphics[width=0.7\textwidth]{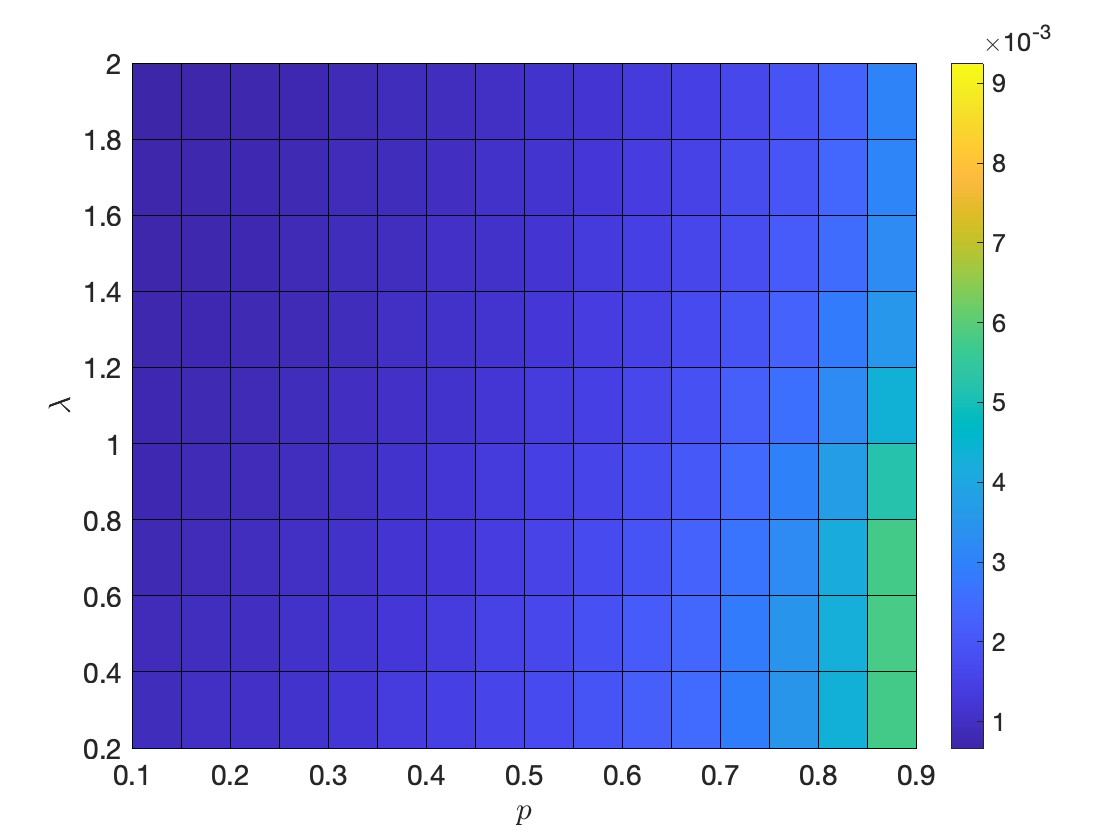}
    \caption{Dataset price $P$ with respect to various $p$ and $\lambda$, for $p\in[0.1,0.9]$ and $\lambda\in[0.2,2]$. Parameters: $\sigma=0.2$, $\mu=0.1$, $x=1$, $\rho = 0.1$, and $r = 0.04$. }\label{fig:priceSample3}
  \end{figure}

\newpage
\section{Conclusions}
In this paper, we developed a pricing model based on the informational value of data assets from the buyer's perspective through the indifference pricing principle, using the magnitude of parameter ambiguity regions to quantify information value of the data assets in the dynamic portfolio selection problems. Without consumption, the price of data assets in the optimal investment problem is strikingly simple, and the consumption will complicate the price. Moreover, we analyze the time-dependent monotonicity of data asset prices, when the utility discount rate is low, prices exhibit monotonic decay over time; conversely, with sufficiently high discount rates, prices first increase and subsequently decrease.

 \end{document}